\documentclass[letterpaper, 10 pt, conference]{ieeeconf}  

\IEEEoverridecommandlockouts                              
\usepackage{graphics} 
\usepackage{amsmath} 
\usepackage{amssymb}  
\usepackage{mathtools}
\usepackage{algpseudocodex,algorithm}
\usepackage{xcolor}
\usepackage{booktabs}
\usepackage[style=ieee,doi=false,url=false,eprint=false,minbibnames=1,maxbibnames=2,isbn=false,date=year]{biblatex}
\newtheorem{remark}{Remark}
\newtheorem{prop}{Proposition}

\newtheorem{assumption}{Assumption}
\newtheorem{corollary}{Corollary}
\newtheorem{definition}{Definition}

\newcommand\copyrighttext{%
  \footnotesize This work has been submitted to the IEEE for possible publication. Copyright may be transferred without notice, after which this version may no longer be accessible.}
\newcommand\copyrightnotice{%
\begin{tikzpicture}[remember picture,overlay]
\node[anchor=south,yshift=10pt] at (current page.south) {\fbox{\parbox{\dimexpr\textwidth-\fboxsep-\fboxrule\relax}{\copyrighttext}}};
\end{tikzpicture}%
}

\title{\LARGE \bf
An Adaptive Control Approach to Treatment Selection for Substance Use Disorders
}

\author{Eric Pulick and Yonatan Mintz
\thanks{E. Pulick and Y. Mintz are with the Department of Industrial and Systems Engineering,  University of Wisconsin - Madison,
        1513 University Ave, Madison, WI, USA, 
        {\tt\small \{pulick\},\{ymintz\}@wisc.edu}}%
}

\begin{document}

\maketitle
\thispagestyle{empty}
\pagestyle{empty}

\begin{abstract}
Despite the massive costs and widespread harms of substance use, most individuals with substance use disorders (SUDs) receive no treatment at all. Digital therapeutics platforms are an emerging low-cost and low-barrier means of extending treatment to those who need it. While there is a growing body of research focused on how treatment providers can identify which patients need SUD support (or when they need it), there is very little work that addresses how providers should select treatments that are most appropriate for a given patient. Because SUD treatment involves months or years of voluntary compliance from the patient, treatment adherence is a critical consideration for the treatment provider. In this paper we focus on algorithms that a treatment provider can use to match the burden-level of proposed treatments to the time-varying engagement state of the patient to promote adherence. We propose structured models for a patient’s engagement over time and their treatment adherence decisions. Using these models we pose a stochastic control formulation of the treatment-provider’s \textit{burden selection problem}. We propose an adaptive control approach that estimates unknown patient parameters as new data are observed. We show that these estimates are consistent and propose algorithms that use these estimates to make appropriate treatment recommendations.
\end{abstract}
\copyrightnotice

\section{INTRODUCTION}

The individual and societal costs associated with substance use, most notably alcohol, nicotine, and opioids, are massive. Globally, the use of these substances accounts for an estimated ten million premature deaths per year \cite{volkowSubstanceUseDisorders2023}. In the United States (US), alone, excessive alcohol use is responsible for roughly \$249 billion in societal costs annually \cite{witkiewitzAdvancesScienceTreatment2019a}. The societal costs of illicit drugs are estimated to be comparable to those of alcohol use \cite{proctorContinuingCareModel2014}.

Fundamental to the long-term treatment of substance use disorders (SUDs), a process called continuing care, are the challenges of patient engagement and adherence \cite{mclellanDrugDependenceChronic2000}. Unlike with acute care, such as emergency overdose treatment, in continuing care patients actively choose whether or not to adhere to treatment recommendations. Thus, it is crucial that a provider understand what types of treatments a given patient is likely to adhere to. The patient's current capacity for engagement, a mental construct that evolves over time, and the burden-level of the recommended treatment are expected to be critical factors in patient adherence decisions \cite{nahum-shaniMobileAssistanceRegulating2021}. While high-burden treatments may effect greater behavioral changes than low-burden treatments, any changes at all rely on the patient's adherence. A successful provider model must adapt recommended treatments to the patient's current capacity for engagement to best support their recovery.

In this work we propose a fully personalized discrete-time model for a patient's engagement capacity and a model for their adherence decisions to recommended treatments. Using these models, we define an adaptive control framework that estimates patient parameters over time and uses these estimates to select effective treatments (specifically to choose the burden-level of the treatment). The authors are not aware of any existing work proposing such models for SUD treatment. This type of structured modeling is particularly motivated by the success of low-dimensional linear models for predicting relapse risk in alcohol use disorder \cite{pulickIdiographicLapsePrediction2025}. The closest work to ours is that of \textcite{ghoshMiWavesReinforcementLearning2024}, which uses a reinforcement learning approach to choose whether or not to send messages to cannabis users to reduce their use. An important difference in our work is the use of fully person-specific models for each patient and the explicit definition of engagement dynamics for use in treatment planning. The remainder of this section provides additional background on SUDs, digital therapeutics platforms and how they can be used for SUD treatment, and an outline for the rest of the paper.

\subsection{Substance Use and Substance Use Disorders}
SUDs, such as alcohol use disorder (AUD), are defined by the \textit{Diagnostic and Statistical Manual of Mental Disorders} (DSM-5) \cite{americanpsychiatricassociationDiagnosticStatisticalManual2013}. Diagnosis is based on whether an individual meets certain problematic characteristics (e.g., substance use that impacts work, school, or other obligations). Estimates suggest that roughly 20 million individuals in the US meet the diagnostic thresholds for an SUD \cite{welty2019substance} and 30\% of the population will meet the diagnostic thresholds for AUD, alone, within their lifetime \cite{grantEpidemiologyDSM5Alcohol2015}. SUDs are typically chronic and relapsing disorders, with one-year return-to-use rates estimated as high as 40\% to 60\% \cite{mclellanDrugDependenceChronic2000,volkowSubstanceUseDisorders2023}.

While SUDs were historically treated as acute conditions, for instance focusing on detoxification, clinicians now advocate treating SUDs like other chronic diseases such as diabetes or hypertension \cite{mclellanDrugDependenceChronic2000,mclellanCanSubstanceUse2013,proctorContinuingCareModel2014}. For these conditions, treatment focuses on long-term monitoring and continuing care (i.e., lower intensity, but sustained treatment efforts) \cite{proctorContinuingCareModel2014,mclellanCanSubstanceUse2013}. Core elements of continuing care for SUDs, particularly AUD, are behavioral interventions (e.g., cognitive behavioral therapy or brief interventions) \cite{witkiewitzAdvancesScienceTreatment2019a,volkowSubstanceUseDisorders2023}. Unfortunately, SUD treatment infrastructure is historically under-resourced  and only a small minority of those with SUDs receive continuing care \cite{sociasAdoptingCascadeCare2016}. Often due to cost and under-availability of resources, recent US estimates suggest that only 10\% of individuals with SUDs receive \textit{any} form of treatment \cite{welty2019substance}. The model we propose in this paper begins to address this gap by providing an effective way for providers to best leverage their resources to maximize patient care.

\subsection{Digital Therapeutics}
The ubiquity of mobile devices has opened up exciting new treatment pathways for SUDs. Digital therapeutics platforms (DTPs), such as A-CHESS \cite{gustafsonSmartphoneApplicationSupport2014c}, can be installed on a patient's smartphone and provide a rich interface between patient and provider. DTPs can both gather relevant patient data and curate treatments, such as behavioral interventions or support tools within the application. While digital interventions are still relatively new, randomized controlled trials suggest that such approaches are at least as effective as existing, clinician-delivered interventions \cite{kilukRandomizedClinicalTrial2018}. Most importantly, DTPs are a low-cost and low-barrier way to extend care to those with SUDs who currently receive no treatment at all.

DTPs can be especially impactful when used to build personalized treatment models. Within medicine, there is a broad recognition of the value of personalizing treatment to the unique needs of each patient \cite{goetzPersonalizedMedicineMotivation2018}. This is echoed specifically in the psychopathology literature, with recent work highlighting how mental health conditions, such as SUDs, can often present or progress differently between patients \cite{wrightPersonalizedModelsPsychopathology2020}. The growing popularity of just-in-time-adaptive-interventions (JITAIs) reflects this need for individualized treatments; JITAIs tailor both the timing and the type/intensity of interventions such that they best serve the time-varying mental state of the patient \cite{nahum-shaniJustinTimeAdaptiveInterventions2018}.

While there are a growing number of successful modeling approaches using mobile-collected SUD patient data to predict \textit{who} needs support or \textit{when} they need it (e.g., for AUD \cite{coughlinJustinTimeAdaptiveIntervention2021,wyantMachineLearningModels2024,pulickIdiographicLapsePrediction2025}), there is little work addressing \textit{which} SUD treatments to recommend to a given patient. Identifying the most effective treatments for a patient is multi-faceted, for instance depending on the cognitive burden of the intervention, its content, its type/format, or its intensity. This paper focuses specifically on how to choose the most appropriate effort- or burden-level of an intervention based on the patient's current capacity for engagement with treatment \cite{nahum-shaniScienceEngagementDigital2024a}.

\subsection{Outline}
Section~\ref{sec:model} introduces our engagement/adherence models and the stochastic control formulation of the provider's $\textit{burden selection problem}$. Section~\ref{sec:params} describes how we estimate unknown patient parameters from observed data and a proof of the estimator's statistical consistency. Section~\ref{sec:algos} introduces two control algorithms for the burden-selection problem. Broadly, these algorithms estimate a patient's unknown parameters from observations and use value iteration to obtain a certainty-equivalence optimal policy. The algorithms differ primarily in how they select treatment actions; one follows the policy given by the maximum likelihood estimates, while the other uses a Thompson-sampling type approach for action selection. Section~\ref{sec:compute} discusses the performance of these algorithms for synthetic patients and practical considerations for their use.

\section{Treatment Model}\label{sec:model}
In this section we discuss our proposed models for treatment adherence and for optimal treatment recommendation. We begin with practical considerations to motivate the assumptions in our proposed treatment adherence model. We next pose the full information stochastic control problem to give insight into the model's structure. Last, we discuss how the engagement dynamics are, in reality, only partially observed and we pose the partial-information control problem.
\subsection{Practical Motivation}
This paper assumes a once-daily interaction model for SUD continuing care. Past work suggests that this quantity of interaction is well-tolerated by patients \cite{moshontzProspectivePredictionLapses2021a,wyantAcceptabilityPersonalSensing2023a}. In this framework, the DTP prompts a patient to complete a brief check-in survey when they wake up. The DTP uses these measurements to identify the particular mental constructs (e.g., stress, craving, etc.) that are driving that patient's risk as part of a personalized risk model \cite{pulickIdiographicLapsePrediction2025}. The DTP then provides the patient with a treatment recommendation for that day. Similar to \cite{nahum-shaniMobileAssistanceRegulating2021}, we categorize possible treatments as either low-burden (e.g., reading a brief reminder from the platform about strategies to mitigate craving) or high-burden (e.g., completing an involved exercise on the DTP for managing craving).

Once given a recommendation by the platform, the patient can either choose to adhere (i.e., complete the treatment) or not.  The provider receives a reward only if the patient adheres to the treatment recommendation (e.g., read the reminder or completed the suggested exercise). In this paper, we assume that the content of a recommendation is suitably personalized to the patient and instead focus on the platform's burden selection problem: each day, the platform must choose between recommending low- or high-burden treatments to best support the patient's recovery.

\subsection{Adherence Model Definition}
With this framework as motivation, we define the following quantities, noting that subscripts denote a quantity's value on a given day, such as day $t$:
\begin{itemize}
    \item $x_t \in \mathcal{X} \subseteq \Re_+$ is the latent state representing the patient's capacity for engagement with the treatment platform;
    \item $a_t \in \mathcal{A} = \{\ell,h\}$ denotes the treatment recommended by the platform, categorized as either low-burden ($\ell$) or high-burden ($h$);
    \item $d_t \in \{0,1\}$ is the binary adherence decision made by the patient;
    \item $c_\ell,c_h\in \Theta^c \subseteq \Re+$ are the adherence costs of the two types of treatment recommendations. In practice, we fix $c_h=1$ to ensure model identifiability;
    \item $\lambda_\ell,\lambda_h \in \Theta^{\lambda} \subseteq \Re_+$ parameterize the probabilities that a patient adheres to a given treatment recommendation;
    \item $b \in \Theta^b \subseteq (0,1)$ denotes a recovery rate for the patient's engagement capacity. 
\end{itemize}

We propose the following discrete-time linear model for the daily evolution of the patient's engagement state:
\begin{equation}
    x_{t+1} = f(x_t,a_t,d_t)= bx_t + c_{a_t}d_t \label{eq:dynamics}
\end{equation}
We assume that the patient's adherence decisions are conditionally independent given the current state and chosen treatment, with $d|x,a\sim\text{Ber}(p_a(x))$. We define the adherence probability as $p_a(x)=e^{-\lambda_{a} x}$. These dynamics imply that the patient's engagement state recovers toward 0 over time per the parameter $b$. Successful adherence carries a corresponding engagement cost ($c_\ell$ or $c_h$) that pushes the state temporarily away from 0. Larger engagement state values imply lower probabilities of adherence. The $\lambda_\ell$ and $\lambda_h$ parameters essentially dictate how quickly adherence probabilities decay from 1 with respect to the engagement state. For brevity, we define $\bar x = \frac{c_h}{1-b}$ and note that these dynamics imply a bounded state space, $\mathcal X = [0,\bar x]$ for $x_0\in[0,\bar x]$.  We believe that the use of linear models is well-motivated by past work; \cite{pulickIdiographicLapsePrediction2025} demonstrated that personalized, low-dimensional linear models are excellent approximations for the mental processes associated with AUD lapse risk. We posit that the same holds true for the mental processes associated with treatment adherence.

\subsection{Full-Information Stochastic Control Problem}
We introduce the full-information problem to provide insight into the model's structure before moving to the partial-information case. The controller makes a treatment recommendation at each time step, $a_t\in\mathcal A$, and receives a treatment reward given by $g(x,a,d)=\gamma_ad$. Note that this reward is non-zero only if a patient adheres to the recommended treatment (i.e., $d=1$). Let $\gamma_\ell$ and $\gamma_h$ denote the treatment rewards associated with the low- and high-burden treatments, respectively. We assume that the treatment rewards are fixed and known. This reflects a real-world situation where expert opinion is available to determine how valuable each treatment is to a patient's recovery. We assume that $\gamma_\ell\leq \gamma_h$, $\lambda_\ell \leq \lambda_h$, and $c_\ell\leq c_h$ such that the treatments reflect a realistic ordering in treatment effects, adherence likelihoods, and adherence costs. Let $\alpha\in(0,1)$ be a discount factor and let $\pi\in\Pi$ represent a policy, $\Pi: \mathcal X \to \mathcal A$ (note that any such policy is admissible as there are no state-related constraints on the action set). The full-information, infinite-horizon stochastic control problem is thus:
\begin{align}
    \underset{\pi\in\Pi}{\text{max}}& \quad \mathbb{E}\left[ \sum_{t=0}^\infty \alpha^t g(x_t,a_t,d_t) \right]\\
    \text{s.t.}&\; x_t = bx_{t-1} + c_{a_{t-1}}d_{t-1} && \forall t\geq 1\nonumber
\end{align}

One might naturally expect this problem to have an optimal threshold policy, namely recommending the high-burden treatment up to a particular threshold and low-burden treatments otherwise. Under such a policy the controller capitalizes on low values of $x$, where adherence probabilities are highest, to try to accrue the larger treatment reward associated with the high-burden treatment. However, our experiments using value iteration (VI) \cite{bertsekas2012dynamic} to compute optimal policies suggest that this is not uniformly true.

Though many model parameterizations display single-threshold policies (i.e., recommending the high-burden treatment until a threshold $\tau$), we found that small parameter variations can lead to counter-intuitive policy structures. Figure~\ref{fig:policy-examples} illustrates this sensitivity by considering optimal policies for three models with identical parameters except for small variations in $c_\ell$. For $c_\ell=0.3$ we observe the usual single-threshold policy structure. However, with $c_\ell=0.2$ we note a double-threshold policy where it is optimal to recommend the \textit{low-burden} treatment for $x\in[0,\tau_1]$, the \textit{high-burden} treatment for $x\in(\tau_1,\tau_2]$, and the low-burden treatment again for $x\in(\tau_2,\bar x]$. With $c_\ell=0.1$ the optimal policy is to always recommend the low-burden treatment. It is important to note that, in practice, these are unknown quantities that must be estimated from observations over time. In this two-action case, the optimal policy is easy to obtain by value iteration. However, this example case underscores that the problem structure is more complex than it may first appear.

\begin{figure}[h!]
    \centering
    \includegraphics[width=0.45\textwidth]{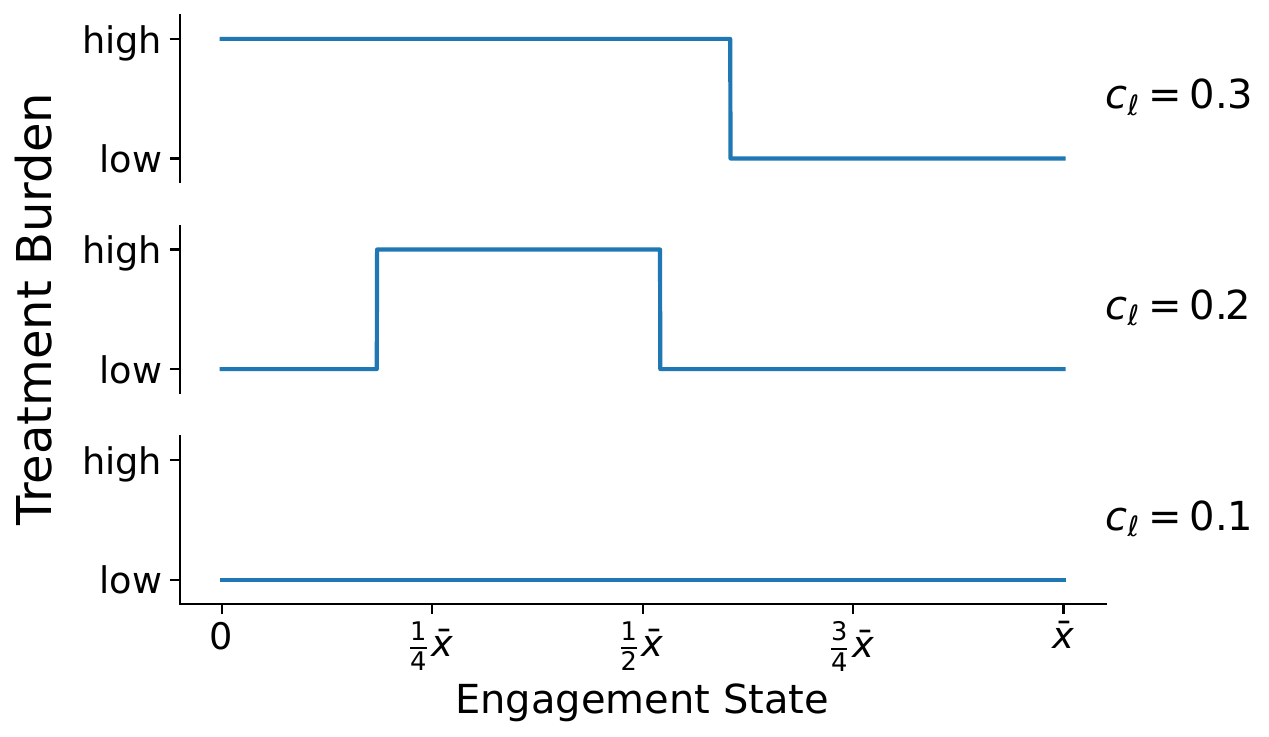}
    \caption{Different optimal policy structures obtained by varying a single parameter ($c_\ell = 0.1, 0.2, 0.3$). Policies were obtained by value iteration. The parameters $(b,\lambda_\ell, \lambda_h, c_h, \gamma_\ell, \gamma_h,\alpha)=(0.9,0.7,0.8,1,0.5,1,0.95)$ were otherwise identical across the three models. }
    \label{fig:policy-examples}
\end{figure}

\subsection{Partial-Information Stochastic Optimal Control}
In reality, we cannot directly observe the patient's latent engagement process. Thus, in the partial-information problem we assume that the structure of the dynamics are as defined in the full information problem, but that $x$ is unobserved and the patient parameters $\theta = (c_\ell,c_h, \lambda_\ell, \lambda_h, b, x_0)$ are fixed but unknown to the controller. We assume that patient adherence decisions, $d$, are always observed, as our motivating setting is a DTP that can measure the patient's adherence behavior (e.g., by confirming that the patient read a message or completed an exercise). Let $\{a_t\}_{t=0}^{n}=(a_0,\dots,a_n)$ and $\{d_t\}_{t=0}^{n}=(d_0,\dots,d_n)$ denote observed sequences of actions and adherence decisions.

\section{Parameter Estimation}\label{sec:params}
In this section we describe the maximum likelihood estimation (MLE) procedure for our problem and prove the statistical consistency of our estimates.

\subsection{Maximum likelihood estimation}
Let the current day be $n+1$. The likelihood function for a set of patient parameters $\theta$, given adherence decisions $\{d_t\}_{t=0}^{n}$ and treatment actions $\{a_t\}_{t=0}^{n}$ is given by:
\begin{multline}
    \mathcal L (\theta|\{d_t\}_{t=0}^n,\{a_t\}_{t=0}^n)=p(\{d_t\}_{t=0}^n|\theta,\{a_t\}_{t=0}^n)=\\
    \prod_{t=0}^np(d_t|x_t,a_t,\theta)\prod_{t=0}^{n-1}p(x_{t+1}|x_{t},d_{t},a_{t},\theta)\label{eq:likelihood-expansion}
\end{multline}
The Bernoulli-defined disturbance gives:
\begin{align}
    p(d_t|x_t,a_t,\theta) &= \left(e^{-\lambda_{a_t}x_t}\right)^{d_t}\left(1-e^{-\lambda_{a_t}x_t}\right)^{1-d_t}
\end{align}
And the dynamics in (\ref{eq:dynamics}) create a degenerate distribution for the state trajectory that can be expressed as constraints:
\begin{align}
    p(x_{t+1}|x_{t},d_{t},a_{t},\theta) &=
    \begin{cases}
    1, & x_{t+1}=bx_{t}+c_{a_{t}}d_{t}\\
    0, & \text{otherwise}
    \end{cases}
\end{align}
Let $x=(x_1,\dots,x_n)$ and recall $\theta=(c_\ell,c_h,\lambda_\ell,\lambda_h,b,x_0)$. Taking the log-likelihood, we can thus obtain $(\hat \theta_n^{\text{MLE}},\hat x^{\text{MLE}})$ as the argmax of the following optimization problem, given sequences $\{d_t\}_{t=0}^n$ and $\{a\}_{t=0}^n$:
\begin{align}
    \underset{\theta\in \Theta,x\in \mathcal X}{\text{max}}\;
    &\sum_{t=0}^n\left[d_t(-\lambda_{a_t}x_t)
        +(1-d_t)\log\left(1-e^{-\lambda_{a_t}x_t}\right)\right]\label{eq:mle}\\
    \text{s.t.}&\;x_{t+1} = bx_{t}+c_{a_{t}}d_{t} \quad t=0,\dots,n-1\nonumber
\end{align}
However, this problem is non-concave in general, containing bilinear terms in both the objective function and the constraints. We address this by considering the $\lambda_\ell$, $\lambda_h$, and $b$ parameters to come from finite sets. 
\begin{prop}
    For fixed $(\lambda_\ell,\lambda_h,b)$, the resultant sub-problem is concave.
\end{prop}
\begin{proof}
    Concavity of the subproblems is more easily seen in the following expanded formulation. Let the current time step be $n+1$ such that we have observations $\{d_t\}_{t=0}^n,\{a_t\}_{t=0}^n$ available. First, we define index sets of the form $A_n^{ a,d}=\left\{t \in \{0,\dots,n\}:a_t= a,d_t = d\right\}$. Note that $A_n^{\ell,0}\cup A_n^{\ell,1}\cup A_n^{h,0}\cup A_n^{h,1}=\{0,\dots,n\}$. For fixed $(\lambda_\ell,\lambda_h,b)$, Problem~\ref{eq:mle} reduces to:
    \begin{align}
        \underset{c_\ell,c_h,x_0,x}{\text{max}}\;
        &\begin{multlined}[t]
            \underbrace{\sum_{t\in A^{\ell,1}_n}-\lambda_\ell x_t + \sum_{t\in A^{h,1}_n} -\lambda_hx_t}_{(a)}\\
            \underbrace{\sum_{t\in A^{\ell,0}_n}\log\left(1-e^{-\lambda_\ell x_t}\right) + \sum_{t\in A^{h,0}_n} \log\left(1-e^{-\lambda_hx_t}\right)}_{(b)}
        \end{multlined}\label{eq:expanded-obj}\\
        \text{s.t.}\; x_{t+1}& = b^tx_0+\sum_{k\in A^{\ell,1}_{t-1}}b^{t-1-k}c_\ell + \sum_{k \in A^{h,1}_{t-1}}b^{t-1-k}c_h, \;\forall t \nonumber\\
        0\leq\; &x_0 \leq \bar x,\quad
        0\leq\;  c_\ell \leq 1,\quad
        c_h = 1\nonumber
    \end{align}
    Here we have expanded the objective function to separately consider each of the four possible outcomes given by the index sets and we have defined the state dynamics in terms of the initial condition $x_0$.

    The (a) terms in are linear in the optimization arguments for fixed $\lambda_\ell,\lambda_h,b$ and are thus concave. The (b) terms are concave by a functional composition argument given by 3.10 in \cite{boyd2004convex}. As concavity is preserved under summation, we have that the objective function is concave.
\end{proof}

\begin{remark}
     Maximum likelihood estimates can thus be found by finding the maximum over the set of $M$ concave subproblems indexed by the corresponding $(\lambda_\ell,\lambda_h,b)$ tuples.
\end{remark}
\begin{remark}
    In practice, $\Theta^\lambda$ and $\Theta^b$ can be made arbitrarily dense to limit the restrictiveness of this finite set assumption, at the expense of additional computational load.
\end{remark}
One additional property that will later be used for proving that our estimators are consistent is the compactness of the parameter space $\Theta$.
\begin{prop}
    The parameter space $\Theta = \Theta^{c_\ell} \times \Theta^{c_h} \times \Theta^{\lambda_\ell} \times \Theta^{\lambda_h} \times \Theta^b \times \mathcal X$ considered during maximum likelihood estimation is compact for any problem iteration $n\geq 1$.
\end{prop}
\begin{proof}
    This is observed by noting that the Problem~\ref{eq:expanded-obj} can be reformulated for any $n\geq 1$ by substituting each $x_t$ term ($t=1,\dots,n$) back into the objective such that it is in terms of only $x_0$ and other problem parameters. Thus, the problem no longer explicitly estimates the trajectory $x=(x_1,\dots,x_n)$ during optimization. The remaining optimization arguments all either come from explicitly finite sets or bounded subsets of $\Re_+$ and are thus compact. As the observations $\{d_t\}_{t=0}^n$ and $\{a_t\}_{t=0}^n$ do not influence these parameter sets, all problem iterations thus optimize over the same compact $\Theta$.
\end{proof}

\subsection{Consistency of the MLE}
Statistical consistency of estimates is critical as it implies the estimates approach their true values as additional data are collected. Here we show that the above maximum likelihood estimation procedure is consistent by proving a broader result for maximum \textit{a posteriori} (MAP) estimation, similar to \cite{aswaniBehavioralModelingWeight2019a}. Broadly, in this approach we introduce an estimator for the parameter posterior likelihood which is shown to be consistent in a Bayesian sense. We then show that MLE and MAP estimators for $\theta$ are recoverable sub-cases of this result.

We begin by applying Bayes' Theorem \cite{bickelMathematicalStatisticsBasic2015} to get an expression for the posterior distribution of $\theta$ given a sequence of observations and actions:
\begin{align}
    p(\theta|\{d_t\}_{t=0}^n,\{a_t\}_{t=0}^n) & = \frac{p(\{d_t\}_{t=0}^n|\theta,\{a_t\}_{t=0}^n)p(\theta)}{Z}
\end{align}
where $Z$ is a normalization constant such that the distribution integrates to 1. As in the MLE case, taking the log-likelihood of this expression gives the MAP estimation problem:
\begin{align}
    \underset{\theta\in \Theta}{\text{max}}
        \log\left(p(\{d_t\}_{t=0}^n|\theta,\{a_t\}_{t=0}^n) \right) + \log\left(p(\theta)\right)-\log(Z) \label{eq:map}
\end{align}
Note that the first term can be expanded per (\ref{eq:likelihood-expansion}) to recover the formulation involving the state dynamics, which can then be included as constraints. For now, though, we consider the problem without this expansion.
\begin{assumption}
   Problem~\ref{eq:map} has a unique argmax.\label{as:identifiability}
\end{assumption}
Without this assumption the subsequent arguments instead show consistency with respect to an equivalence class of parameters. In practice, though, we found that fixing $c_h$ to 1 was sufficient to satisfy this assumption.
\begin{assumption}
    The prior $p(\theta)$ is jointly log-concave in $c_\ell,c_h,x_0$. Additionally, $p(\theta) > 0$ for all $\theta \in \Theta$.\label{as:prior-nice}
\end{assumption}
As most probability distributions satisfy log-concavity \cite{boyd2004convex}, the first part of the assumption is quite mild. It ensures that Problem~\ref{eq:map} remains concave when broken into subproblems indexed by finite values of $(\lambda_\ell,\lambda_h,b)$, as in Problem~\ref{eq:expanded-obj}. The second part of the assumption is needed to ensure the true parameters are recoverable for any $\theta\in\Theta$. Importantly, the uniform prior distribution, which is represented as a constant function for all $\theta\in\Theta$ trivially satisfies both parts of this assumption.

Note that the $\log(Z)$ term in Problem~\ref{eq:map} scales the posterior likelihood, but is a constant and therefore can be removed without impacting the argmax. Let $i=1,\dots,M$ index the subproblems, each fixing some $(\lambda_\ell^i,\lambda_h^i,b^i)$. The $i^{th}$ subproblem is thus:
\begin{align}
    \underset{c_\ell,c_h,x_0}{\text{max}}\;
    &\begin{multlined}[t]
        \log\left(p(\{d_t\}_{t=0}^n|\theta,\{a_t\}_{t=0}^n) \right)
        + \log\left(p(\theta)\right) \label{eq:map-subproblem}
    \end{multlined}\\
    \text{s.t.}&;\ \lambda_\ell=\lambda_\ell^i,\;\lambda_h=\lambda_h^i,\;b=b^i\nonumber\\
    &\; 0\leq c_\ell\leq1,\;c_h=1,\;x_0\in[0,\bar x]\nonumber
\end{align}
Let $(\hat c_\ell^i,\hat c_h^i,\hat x_0^i)$ be the argmax of the $i^{th}$ subproblem. Recall $\theta=(c_\ell,c_h,\lambda_\ell,\lambda_h,b,x_0)$ and let $\psi(\theta)$ denote the evaluation of the objective function in Problem~\ref{eq:map-subproblem} for a given $\theta$. Then we have the argmax of the $i^{th}$ subproblem is $\hat \theta^i_n=(\hat c_\ell^i,\hat c_h^i,\lambda_\ell^i,\lambda_h^i,b^i,\hat x_0^i)$ and the corresponding solution value is $\psi(\hat \theta_n ^i)$. The MAP estimate for $\theta$, $\hat\theta^{\text{MAP}}_n$, is thus straightforward to obtain as $\text{max}_{i=1,\dots,M}\psi(\hat \theta^i_n)$.

However, to show our estimates are consistent in a Bayesian sense, we must obtain an estimate for the  the posterior likelihood $p(\theta|\{d_t\}_{t=0}^n,\{a_t\}_{t=0}^n)$. In theory, this could be obtained directly by evaluating Problem~\ref{eq:map}, but because $Z$ is typically difficult to calculate, we propose an alternative estimator. Specifically, we note that our iterated subproblems provide a collection of profile likelihood estimates \cite{murphyProfileLikelihood2000a}, including the MAP estimate, that we can use to estimate our uncertainty in $\theta$. We propose using the estimator:
\begin{equation}
    \hat p(\theta|\{d_t\}_{t=0}^n,\{a_t\}_{t=0}^n) = \frac{e^{\psi(\theta)}}{\sum_{i=1}^Me^{\psi(\hat \theta^i)}}\label{eq:p-estimator}
\end{equation}
As the MAP estimator $\hat\theta_n^{\text{MAP}}$ is necessarily one of the elements in the denominator, we have the useful property that $\hat p(\theta|\{d_t\}_{t=0}^n,\{a_t\}_{t=0}^n) \in [0,1]$. We next show that this estimator is consistent in a Bayesian sense \cite{bickelMathematicalStatisticsBasic2015}.
\begin{definition}
    The posterior estimate (\ref{eq:p-estimator}) is consistent if for any value of the patient's true parameters $\theta^*\in\Theta$, and constants $\epsilon,\delta>0$, we have that $\mathbb P_0\left(\hat p(\mathcal{B}'_\delta(\theta^*)|\{d_t\}_{t=0}^n,\{a_t\}_{t=0}^n)\geq \epsilon\right)\to 0$ as $n\to\infty$.
    $\mathbb{P}_0$ is the probability law given under $\theta^*$, $\mathcal{B}'_\delta(\theta^*) = \{\theta \not \in \mathcal B_\delta(\theta^*)\}$, and $\mathcal{B}_\delta(\theta^*)$ is an open $\delta$-ball around $\theta^*$.
\end{definition}
This means that a consistent posterior likelihood estimate behaves such that all probability mass concentrates within an arbitrarily small $\delta$-ball surrounding the true parameters. The final assumption required to show consistency for our posterior likelihood estimator is one of \textit{sufficient excitation}:
\begin{assumption}
    Let the patient's true parameters be $\theta^*\in\Theta$. The controller selects actions such that for any $\theta\in\Theta$ and $\delta>0$ almost surely we have:
    \begin{align}
        \underset{\mathcal B'_\delta(\theta^*)}{\text{max}}\lim_{n\to\infty} \sum_{t=0}^n\log \frac{p(d_t|\tilde x_t(\theta),a_t,\theta)}{p(d_t|\tilde x_t(\theta^*),a_t,\theta^*)}=-\infty
    \end{align}
    where $\tilde x(\theta) = (\tilde x_1,\dots,\tilde x_n)$ is the induced state trajectory under $\theta$ (i.e., the states found when solving the version of Problem~\ref{eq:map} with the dynamics as constraints).\label{as:sufficient-excitation}
\end{assumption}
This type of assumption ensures that different combinations of states and actions occur sufficiently often for proper parameter identification. It is a mild assumption that is common in the adaptive control literature \cite{astromAdaptiveControl1991}. A sufficient means for the controller to satisfy this condition is through the inclusion of a small amount of random noise in action selection \cite{bitmeadPersistenceExcitationConditions1984}.

\begin{prop}
    If Assumptions~\ref{as:identifiability}-\ref{as:sufficient-excitation} hold then the posterior likelihood estimator $\hat p(\theta|\{d_t\}_{t=0}^n,\{a_t\}_{t=0}^n)$ is consistent.\label{prop:post-consistency}
\end{prop}
\begin{proof}
    We follow a similar approach to \cite{mintzBehavioralAnalyticsMyopic2023}. First let $\theta^*$ be the patient's true parameters. First observe that:
    \begin{multline}
        \log(\hat p(\theta|\{d_t\}_{t=0}^n,\{a_t\}_{t=0}^n))\\
        =\log\left(\hat p(\theta^*|\{d_t\}_{t=0}^n,\{a_t\}_{t=0}^n) \frac{\hat p(\theta|\{d_t\}_{t=0}^n,\{a_t\}_{t=0}^n)}{\hat p(\theta^*|\{d_t\}_{t=0}^n,\{a_t\}_{t=0}^n)}\right)\\
        =\log\left(\hat p(\theta^*|\{d_t\}_{t=0}^n,\{a_t\}_{t=0}^n)\right) + \log\left(\frac{p(\theta)}{p(\theta^*)}\right)\\
        + \sum_{i=1}^n\log\left(\frac{p\left(\tilde x_t(\theta)|\tilde x_{t-1}(\theta),a_{t-1},d_{t-1},\theta\right)}{p\left(\tilde x_t(\theta^*)|\tilde x_{t-1}(\theta^*),a_{t-1},d_{t-1},\theta^*\right)} \right)\\
        + \sum_{i=0}^n\log\left(\frac{p\left(d_t|\tilde x_t(\theta),a_t,\theta\right)}{p\left(d_t|\tilde x_t(\theta^*),a_t,\theta^*\right)} \right)
    \end{multline}
    where the final step follows from separating logarithm terms, noting that by the definition of $\hat p$ in (\ref{eq:p-estimator}) the denominator terms cancel in ratios of $\hat p$, and expanding the likelihood function per (\ref{eq:likelihood-expansion}). Moving through the terms in order we note that, by construction, $\log\left(\hat p(\theta^*|\{d_t\}_{t=0}^n,\{a_t\}_{t=0}^n)\right)$ is bounded above by 0 and that the ratio of the priors is a constant. The $p(\tilde x_t(\theta)|\cdot$) quantities in the third term are all 1 as they are degenerate distributions for the state trajectory, making the collective summation 0. Using Assumption~\ref{as:sufficient-excitation}, we have that $\text{max}_{\theta \in \mathcal B'_\delta(\theta^*)}\log(\hat p(\theta|\{d_t\}_{t=0}^n,\{a_t\}_{t=0}^n))\to -\infty$ for any $\delta>0$ almost surely and equivalently $\text{max}_{\theta \in \mathcal B'_\delta(\theta^*)}\hat p(\theta|\{d_t\}_{t=0}^n,\{a_t\}_{t=0}^n)\to 0$ almost surely. Note that this implies point-wise convergence. To show uniform convergence we use a volume bound:
    \begin{multline}
        \hat p(\mathcal{B}'_\delta(\theta^*)|\{d_t\}_{t=0}^n,\{a_t\}_{t=0}^n)\\ =\int_{\mathcal{B}'_\delta(\theta^*)}\hat p(\theta|\{d_t\}_{t=0}^n,\{a_t\}_{t=0}^n)\; d\theta\\
        \leq \text{volume}(\Theta)\cdot\underset{\theta \in \mathcal{B}'_\delta(\theta^*)}{\text{max}}\hat p(\theta|\{d_t\}_{t=0}^n,\{a_t\}_{t=0}^n) \to 0 \label{eq:posterior-final}
    \end{multline}
    completing the proof as (\ref{eq:posterior-final}) holds almost surely for any $\delta>0$.
\end{proof}
\begin{corollary}
    Both the MAP and MLE estimators ($\hat\theta^{\text{MAP}}_n$,$\hat\theta^{\text{MLE}}_n$) converge in probability to $\theta^*$ as $n\to \infty$.\label{cor:estimator-consistency}
\end{corollary}
\begin{proof}
    We define the events 
    $E_1=\{\hat\theta_n^{\text{MAP}}\not \in \mathcal{B}(\theta^*,\delta)\}$ and 
    $E_2=\{\text{max}_{\theta \in \mathcal{B}'_\delta(\theta^*)} \hat p(\theta|\{d_t\}_{t=0}^n,\{a_t\}_{t=0}^n) \geq 
            \text{max}_{\theta \in \mathcal{B}_\delta(\theta^*)} \hat p(\theta|\{d_t\}_{t=0}^n,\{a_t\}_{t=0}^n)\}$
    , with $\mathcal{B}_\delta(\theta^*)$ and $\mathcal{B}'_\delta(\theta^*)$ defined as before. Note that $E_1 \subset E_2$, implying that $\mathbb{P}(E_1)\leq\mathbb{P}(E_2)$. By Proposition~\ref{prop:post-consistency} we have that $\mathbb{P}(E_2)\to 0$ as $n\to \infty$, implying $\mathbb{P}(E_1)\to 0$ and $\hat\theta^{\text{MAP}}_n\overset{p}{\to}\theta^*$. Further note that we recover the MLE problem under a uniform prior, which we noted satisfies Assumption~\ref{as:prior-nice}, so $\hat\theta^{\text{MLE}}_n\overset{p}{\to}\theta^*$.
\end{proof}

\section{Proposed Control Algorithms}\label{sec:algos}
For this problem, it was computationally inexpensive to obtain optimal policies using value iteration (VI). We propose two certainty-equivalence-type algorithms reliant on VI, but using different action-selection schemes. In Algorithm~\ref{alg:eps-greedy-alg}, we estimate the model parameters each day using the MLE approach described in Section~\ref{sec:params}. We obtain a policy $\pi^{\text{MLE}}_n$ using VI with model parameters $\hat \theta_n^{\text{MLE}}$. To induce sufficient excitation, we explore the sub-optimal action given by $\pi_n^{\text{MLE}}$ per a decaying exponential with parameter $\beta$. Algorithm~\ref{alg:thompson-alg} also estimates parameters at each time step following the procedures in Section~\ref{sec:params}, but selects actions using a Thompson-sampling approach. The algorithm samples $\hat\theta^i_n$ from the $M$ subproblem solutions with probabilities $\hat p(\hat\theta_n^i|\{d_t\}_{t=0}^n,\{a_t\}_{t=0}^n)$. The algorithm uses VI with the sampled $\hat\theta_n^i$ to calculate $\pi^{i}_n$ and selects the optimal action under $\pi^i_n$. While this action selection will satisfy sufficient excitation in the limit, we noted better finite-time performance for Algorithm~\ref{alg:thompson-alg} with a brief initialization period to improve early estimates for $\hat p(\theta^i)$. 
\floatname{algorithm}{Algorithm}
\begin{algorithm}
\begin{algorithmic}[1]
    \State Alternate actions $\ell$ and $h$ for $n=1,2$
    \ForAll{$n > 2$}
        \State Estimate $\hat \theta_n^{\text{MLE}}$
        \State Obtain optimal policy $\pi_n^{\text{MLE}}$ by VI
        \State Sample $\epsilon\sim U(0,1)$
        \If{$\epsilon\leq e^{-\beta n}$}
            \State Take the sub-optimal action per $\pi_n^{\text{MLE}}$
        \Else
            \State Take the optimal action per $\pi_n^{\text{MLE}}$
        \EndIf
    \EndFor
\end{algorithmic}
\caption{MLE-$\beta$ Controller}
\label{alg:eps-greedy-alg}
\end{algorithm}
\floatname{algorithm}{Algorithm}
\begin{algorithm}
\begin{algorithmic}[1]
    \State Alternate taking actions $\ell$ and $h$ for $n=1,\dots,T$
    \ForAll{$n > T$}
        \State Estimate $\hat \theta^i_n$ per Problem~\ref{eq:map-subproblem}, $i=1,\dots,M$
        \State Select $\hat\theta^i$ with probability $\hat p(\hat\theta_n^i|\{d_t\}_{t=0}^n,\{a_t\}_{t=0}^n))$
        \State Obtain optimal policy $\pi^i_n$ by VI
        \State Take optimal action per $\pi^i_n$
    \EndFor
\end{algorithmic}
\caption{Thompson Sampling Controller}
\label{alg:thompson-alg}
\end{algorithm}

\section{Computational Results}\label{sec:compute}
\begin{figure*}[t]
    \centering
    \includegraphics[width=0.95\textwidth]{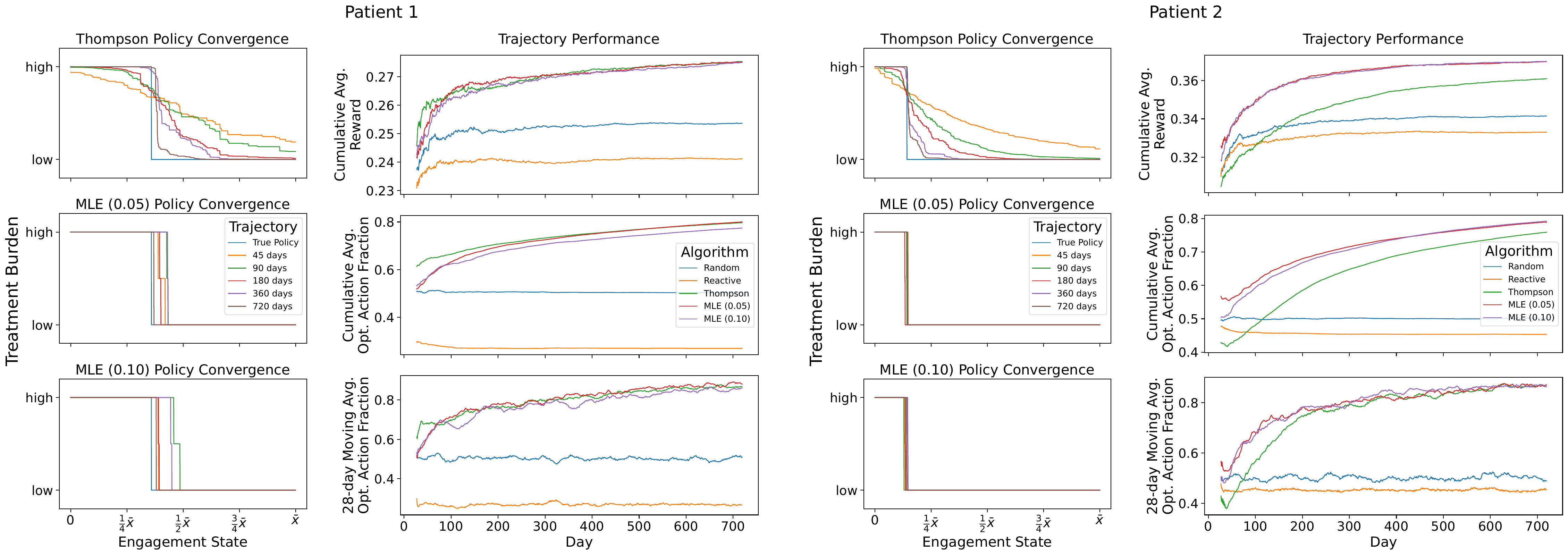}
    \caption{Performance of the tested algorithms, separated for Patient 1 (left) and Patient 2 (right). For each patient, the left three plots show the optimal policies of each algorithm after trajectories of increasing length. Policy curves are the median values across replications. For each patient, the right three plots show performance summaries for 720-day trajectories, given as the mean for each algorithm across replications. Vertically stacked plots share x-axes.}
    \label{fig:compute-results}
\end{figure*}
In this section we describe the results of our computational experiments for two synthetic patients. Computations were performed in Python using a laptop computer with a 3.2GHz processor and 16GB of RAM. The true model parameters for the synthetic patients are provided in Table~\ref{tab:model-params}.
\begin{table}[b]
\centering
\caption{True parameters for synthetic patients 1 and 2.}
\begin{tabular}{ccc} \toprule
Parameter & Patient 1 & Patient 2  \\ \midrule
Adherence costs $(c_\ell,c_h)$ & (0.7,1) & (0.1,1)\\
Adherence parameters $(\lambda_\ell,\lambda_h)$ & (0.4,1) & (0.2,1)\\
Recovery parameter $(b)$ & 0.8 & 0.8 \\
Initial state $(x_0)$ & 0.5$\bar x$ & 0.5$\bar x$\\
Adherence rewards $(\gamma_\ell,\gamma_h)$ & (0.5, 1) & (0.4, 1)\\
Discount factor $(\alpha)$ & 0.95 & 0.95\\
\bottomrule
\end{tabular}\label{tab:model-params}
\end{table}
The patients share numerous parameters, for example the high-burden treatment cost $c_h$, high-burden adherence parameter $\lambda_h$, and recovery parameter $b$. Note that identical $b$ and $c_h$ quantities imply an identical bounded state space between the patients (recall $\bar x = \frac{c_h}{1-b}$). However, there are a few key differences to note. Namely, Patient 2 has a substantially lower cost associated with the low-burden treatment ($c_\ell$) than Patient 1. Patient 2 also has a smaller low-burden adherence parameter ($\lambda_\ell$). Recall that the functions that generate adherence probabilities are decaying exponentials, so lower values of $\lambda$ correspond to flatter decay curves over the engagement state space. Adherence rewards are similar between the two individuals. We chose these model parameters to reflect two plausible real-world cases; Patient 1 represents an individual for whom both intervention-types are relatively costly, while Patient 2 represents an individual who perceives a large difference between low- and high-burden treatment types. Both patients displayed single-threshold optimal policies, but Patient 2's threshold was lower than that of Patient 1. 

To evaluate performance, we simulated trajectories of 45, 90, 180, 360, and 720 days using our proposed algorithms. As Algorithm~\ref{alg:eps-greedy-alg} uses an epsilon-greedy exploration approach, we ran versions of the algorithm with $\beta$ values of 0.05 and 0.10, referred to as `MLE (0.05)' and `MLE (0.10),' respectively. We found that Algorithm~\ref{alg:thompson-alg}, referred to as `Thompson,' had better performance with a short random initialization period ($T=10$). We benchmarked against two heuristic approaches, `Random' and `Reactive'. `Random' took low- and high-burden treatment actions with equal probability while `Reactive' naively chose the high-burden treatment if the patient adhered in the previous time step and the low-burden treatment otherwise. Note that VI used a 3000-step discretization of the state space and a value function convergence tolerance of 0.001. Experiments considered subproblems indexed by tuples built from finite sets $\Theta^{\lambda}=\{0.2,0.4,0.6,0.8,1\}$ and $\Theta^b=\{0.6, 0.7, 0.8,0.9\}$ for both patients.

Algorithm performance from 100 replications for each patient and algorithm is summarized in Figure~\ref{fig:compute-results}. These figures show the policies for each algorithm after fixed-length trajectories. The MLE algorithms have the expected step-shape as they come from a single policy using the maximum likelihood parameters. The Thompson algorithm is shown using a weighted average of the policies given by each possible subproblem parameter in its uncertainty distribution. The purpose of these plots is to confirm that the optimal policies approach the true optimal policy for the individual as the trajectory length grows. Indeed, we see this to be the case for both patients under Algorithms~\ref{alg:eps-greedy-alg} and \ref{alg:thompson-alg}.

Figure~\ref{fig:compute-results} also contains three summaries of 720-day trajectory performance. First is the cumulative average reward obtained by the controller, measured as the average reward up to each time step. Next is the cumulative average fraction of the controller's actions that are optimal, measured against an oracle with access to the patient's true parameters. Last is a 28-day moving average of the optimal action fraction, giving insight into local performance during the trajectory.

Broadly, both of our proposed algorithms meaningfully outperformed the heuristic benchmarks. For both patients, we observe that MLE (0.05) generally outperformed MLE (0.10); MLE (0.05) has a slower-decaying exploration curve, highlighting the value of a longer period of exploration in this context. Additionally, the 28-day averages for optimal action fraction become quite similar across algorithms within approximately 200 days, suggesting a general equivalence of algorithms in the long run. 

However, we see important differences between the algorithms in the early portions of the trajectory. In particular, the early-trajectory Thompson algorithm performance for Patient 2 is actually worse than Random and Reactive. This is especially interesting compared to Patient 1, where Thompson shows the best early-trajectory performance across all methods. We suspect that a reason for this difference is the location of the optimal threshold relative to typical engagement state values observed during the patient's trajectory. The position of Patient 1's threshold meant that the optimal action was frequently changing between low- and high-burden treatments, while the low-burden action was typically optimal for Patient 2. Post-hoc analysis suggested Thompson's poor exploration for Patient 2 was due to premature posterior collapse. Thus, while the Thompson sampling approach shows promising exploration performance for some patients (e.g., Patient 1) this may not extend to all other patient types (e.g., Patient 2). We suspect that this is because uncertainty in subproblem parameters is profiled out during estimation, particularly $c_\ell$, which was shown to have important effects on the optimal policy structure in Figure~\ref{fig:policy-examples}. In future work we plan to address this with new approaches to incorporate subproblem parameter uncertainty into action selection. This could be done, for instance, by solving additional subproblems that consider restricted domains for parameters like $c_\ell$. This would give a richer representation of the posterior at the expense of additional computational burden. Another approach would be the inclusion of priors, for instance from randomized controlled trials, which is known to improve Thompson-type exploration \cite{zhouEvaluatingMachineLearning2018}. Because we rely on synthetic patient data, we chose not to incorporate priors into estimation (though we know MAP estimators to be consistent by Corollary~\ref{cor:estimator-consistency}). We expect that representative prior information could drastically improve the early-trajectory performance of both of our algorithm types. 

\section{Conclusions}
In this paper we proposed a novel adaptive control framework for SUD treatment selection. We introduced structured adherence and engagement models for patients and the corresponding stochastic control problem for the provider. We proposed an estimation procedure to identify the patient's unknown parameters and proved the statistical consistency of these estimates. We provided two certainty-equivalence algorithms for treatment selection and evaluated their performance on synthetic patients. We note that the underlying treatment selection problem appears to have interesting structural properties, even in this two-action setting. We plan to explore these structural properties under broader action sets in future work.



\printbibliography

\end{document}